\documentclass[11pt]{article}

\usepackage[utf8]{inputenc}
\usepackage[T1]{fontenc}
\usepackage{lmodern}
\usepackage{microtype}
\usepackage{geometry}
\usepackage{amsmath,amssymb,amsthm,mathtools}
\usepackage{booktabs}
\usepackage{array}
\usepackage{enumitem}
\usepackage{hyperref}
\usepackage{url}
\usepackage{graphicx}
\usepackage{float}
\usepackage{caption}
\graphicspath{{./}}

\hypersetup{hidelinks}

\newtheorem{definition}{Definition}[section]
\newtheorem{assumption}[definition]{Assumption}
\newtheorem{lemma}[definition]{Lemma}
\newtheorem{theorem}[definition]{Theorem}

\newcommand{\Prb}{\Pr}
\newcommand{\Hinf}{H_{\infty}}
\newcommand{\eps}{\varepsilon}
\newcommand{\bits}{\{0,1\}}
\newcommand{\concat}{\mathbin\Vert}
\newcommand{\pos}[1]{\left[#1\right]_{+}}

\title{\textbf{Contextual Chain: Lightweight Continuity Authentication for Intermittently Connected Devices}\\
\large Resource-Based Security from Shared Experience}
\author{Song-Ju Kim}
\date{October 2026}

\begin{document}
\maketitle

\begin{abstract}
  Can authentication make memory, rather than computational hardness, the attacker's bottleneck? Contextual Chain is a lightweight continuity protocol for intermittently connected devices that share evolving physical or operational context. An honest device follows one realized history, updating a compact accumulator and fixed hash-based readiness lanes; outages cause pause or bounded rollback, not branch search. After the epoch is frozen, a fresh challenge selects one lane under a short deadline. An outsider that missed context may therefore need to prepare for many mature histories before learning which one will be tested. In the standard random-oracle model, a causal counting theorem lower-bounds the deadline-accessible retained state required for a target success probability against arbitrary nonlinear preselection encoding and adaptive post-selection queries, accounting for sequential depth, candidate queries, and cross-target protected information obtained online. Honest readiness memory remains fixed and independent of the number of plausible histories. Contextual Chain thus converts shared-experience uncertainty into a tunable preparation requirement without transferring combinatorial complexity to lightweight devices.
\end{abstract}

\section{Introduction}
Can a device that repeatedly loses its link still prove that it has remained on the same shared trajectory?  This question arises naturally in intermittently connected IoT and edge systems, where association may persist across outages even though one party can miss part of the physical or operational context.  Re-establishing identity with a long-term key does not by itself show continuity of that shared experience.  Contextual Chain is designed for this second problem.

The honest protocol is intentionally small.  A synchronized device maintains one contextual accumulator and $J$ fixed-width readiness lanes.  Each accepted shared event advances that single realized state; failed reconciliation or link loss causes a pause, abort, or bounded rollback rather than a search over alternatives.  At the end of an epoch the state is frozen, the verifier chooses one lane only then, and a fresh nonce requests a response under a short deadline.  The continuity step is therefore hash-based, uses fixed memory, and requires no candidate enumeration or public-key operation.  With the illustrative choice $J=256$ and $\ell=256$ bits, the readiness vector occupies 8~KiB.

The security asymmetry appears only on the uncertain side.  An outsider that missed part of the accepted context may face a posterior over many plausible histories.  Recent uncertainty is not credited if it can be regenerated after the challenge.  Contextual Chain therefore uses only a \emph{mature} prefix lying beyond the calibrated post-selection regeneration horizon.  Streaming readiness turns every mature candidate history into a protected readiness chain, while fresh lane selection occurs only after preparation has stopped.  To answer in time, the attacker must either carry useful state across that selection cut or obtain enough protected information after selection within the calibrated deadline resources.

Figure~\ref{fig:protocol-overview} separates the device implementation from the adversarial preparation problem.  The upper row is the complete honest path: accepted events update one state, interruptions never create honest branches, and reconnection leads to one frozen state and one randomly selected lane.  The lower fan-out is not memory held by the legitimate device; it is the attacker's posterior over mature histories when shared experience is incomplete.  The vertical selection cut is the key timing boundary: information on its left must have been prepared before the challenge, whereas information on its right must be regenerated within $\Delta$.  This single-path honest workload versus multi-target adversarial preparation is the core mechanism of the protocol.

\begin{figure}[H]
\centering
\includegraphics[width=0.97\linewidth]{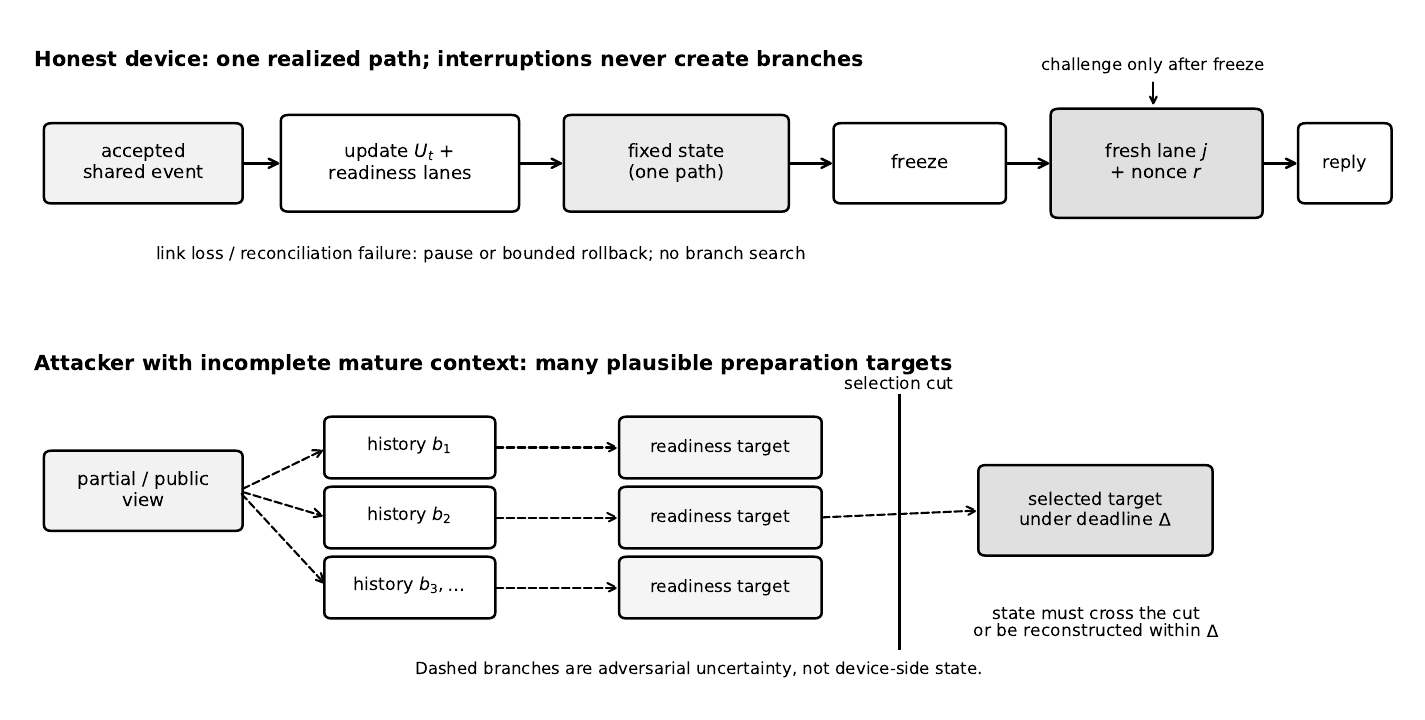}
\caption{Operational view of Contextual Chain.  The honest device advances one realized trajectory and keeps a fixed readiness vector; link interruption or reconciliation failure causes pause or bounded rollback, not candidate branching.  An attacker with incomplete mature context may instead face many plausible histories and readiness targets.  The challenge is sampled only after the state is frozen.  Solid lines denote the realized protocol path, dashed branches denote adversarial uncertainty, and the vertical cut separates pre-challenge preparation from work that must fit within the deadline $\Delta$.}
\label{fig:protocol-overview}
\end{figure}

The security claim is resource-bounded but strategy-independent.  The adversary may retain any nonlinear encoding before selection and may make adaptive post-selection random-oracle queries, including queries associated with other plausible targets.  For a desired success probability, the theorem lower-bounds deadline-accessible retained state as a function of exact mature posterior coverage, post-selection dependent depth, candidate-query capacity, and protected cross-target information acquired online.  Thus the proof does not enumerate known attacks: any strategy inside these explicit resource bounds is covered by the same counting argument.

The idea belongs to the broader family of resource-bounded cryptography, including bounded-storage, proof-of-space, memory-hard, and preprocessing/advice models~\cite{Maurer1992,DamgardEtAl2008,DodisQuachWichs2022,DziembowskiEtAl2015,AlwenSerbinenko2015,AlwenBlockiPietrzak2018,CorettiDodisGuoSteinberger2018,GhoshalTessaro2023,DryjaLiuPark2018}.  Contextual Chain differs in how the resource pressure is created: shared-experience uncertainty generates an epoch-specific family of mature readiness targets, and a fresh post-selection challenge decides which target matters only after preparation.  The transformation is
\[
\boxed{\begin{gathered}
\text{shared contextual uncertainty}\rightarrow\text{maturation}\rightarrow\text{streaming readiness}\\
\rightarrow\text{post-selection test}\rightarrow\text{retained-state requirement}.
\end{gathered}}
\]
The first Contextual Chain study introduced the operational principle of using evolving context to guide continuation and evaluated a recovery-oriented realization for intermittently connected networks~\cite{KimContextualChainSystems2026}.  The Version~1 specification documented that initial operational protocol and treated resource-sensitive proof of context as a non-normative experimental extension~\cite{KimCCSpecV1}.  Version~2 generalized Contextual Chain into a design family and deliberately left cryptographic authentication, adversary models, and formal resource lower bounds to realization-specific profiles~\cite{KimCCSpecV2}.  The construction developed here provides such a realization: it instantiates continuity authentication and gives a quantitative causal retained-state theorem.

The resulting separation is the main contribution: the honest endpoint runs a short fixed-memory state machine, while the proof quantifies how contextual uncertainty can force a much larger preparation state on an uncertain attacker.  Section~2 gives the protocol, Section~3 defines the resource-bounded adversary, Section~4 turns accepted shared experience into mature posterior coverage, and Section~5 proves the retained-state bound.  Sections~6--8 cover honest cost, calibration, and positioning; the appendices collect proof details, deployment extensions, calibration guidance, and reproducible numerical examples.

\section{Protocol Architecture}

\subsection{Protocol at a glance}
The honest device maintains only the confirmed state $S_k$, one contextual accumulator $U_t$, the $J$ readiness lanes $(X_{t,1},\ldots,X_{t,J})$, and bounded reconciliation/transaction metadata.  One epoch consists of four operations:
\begin{enumerate}[leftmargin=2em,itemsep=0.25em]
\item reconcile the next shared observation to one canonical accepted event $E_t$;
\item update $U_t$ once and advance every readiness lane by $L$ dependent hash steps;
\item after $T$ accepted events, freeze the state and sample the challenge lane only then;
\item answer from that lane under the deadline and commit the next confirmed state only after successful mutual authentication.
\end{enumerate}
No honest-side branch search occurs.  The readiness memory is exactly $J\ell$ bits; for $J=256$ and $\ell=256$ this is 8~KiB, plus the context, confirmed state, and reconciliation metadata.  The per-event readiness work is $JL$ evaluations and the response uses one selected lane.  Section~\ref{sec:honest-cost} gives the resulting feasibility envelope.

For the analysis, the domain-separated random oracles have explicit output lengths
\[
\begin{aligned}
H_{\rm ctx}&:\bits^*\to\bits^\lambda,
& H_{\rm rd}&:\bits^*\to\bits^\ell,\\
H_{\rm resp}&:\bits^*\to\bits^{\kappa_{\rm resp}},
& H_{\rm root},H_{\rm next}&:\bits^*\to\bits^{\lambda}.
\end{aligned}
\]
The state-transition hashes used only after successful authentication may be domain-separated instances of the same implementation hash; they are not part of the retained-state lower-bound game.  The device-facing configuration consists of the source/reconciliation adapter, the hash/authenticator choices, and $(J,L,T,\Delta)$.  Quantities such as $d_\Delta$, $Q_\Delta$, and $b_\Delta$ are attacker-side calibration parameters used to certify a deployment; they do not add messages, branches, or retained structures to the honest protocol.

\subsection{Notation for the analysis}
\begin{center}
\begin{tabular}{p{0.20\linewidth}p{0.72\linewidth}}
\toprule
Symbol & Meaning \\
\midrule
$E_t$ & canonical accepted shared event; \\
$U_t$ & contextual accumulator after event $t$; \\
$X_{t,j}$ & readiness state of lane $j$; \\
$J,\ell,w=J\ell$ & lane count, lane width, and honest readiness width; \\
$L$ & sequential readiness depth added per accepted event; \\
$\Delta$ & post-selection response window; \\
$d_\Delta$ & maximum dependent readiness depth executable within $\Delta$; \\
$Q_\Delta$ & maximum target-bearing candidate-query budget in one selected-target reconstruction within $\Delta$; \\
$V_{\rm pre}^{\rm free}$ & zero-cost adversarial view available before challenge selection; \\
$\tau_\Delta$ & mature-prefix frontier; \\
$M_\Delta$ & log-cardinality of deadline-accessible state retained across the selection cut; \\
$b_\Delta(A)$ & log-cardinality of residual cross-target protected query/response transcripts for target set $A$; \\
$K_\beta$ & number of mature branch--lane targets required to cover mass $\beta$; \\
\bottomrule
\end{tabular}
\end{center}

\subsection{Confirmed state and one realized transactional trajectory}

At the start of epoch $k$, legitimate participants $A$ and $B$ share a confirmed state
\[
S_k\in\bits^{\lambda}.
\]
The protocol generates a fresh attempt identifier $a_k$ and public salt $\sigma_k$. The pending contextual accumulator is initialized by a domain-separated hash,
\begin{equation}
U_0
=H_{\rm ctx}(\texttt{ACC},S_k,k,a_k,\sigma_k).
\label{eq:init-context}
\end{equation}
Each readiness lane is initialized independently in the analysis domain,
\begin{equation}
X_{0,j}
=H_{\rm rd}(\texttt{INIT},S_k,k,a_k,\sigma_k,j),
\qquad j=1,\ldots,J.
\label{eq:init-readiness}
\end{equation}
The honest participant maintains one realized $U_t$ and one realized vector $(X_{t,1},\ldots,X_{t,J})$. It never enumerates candidate histories.

\subsection{Canonical accepted shared events}

A deployment-specific bounded reconciliation procedure maps noisy local observations to a canonical event symbol or failure:
\[
(E_t^A,E_t^B,\mathcal H_t)
\leftarrow
\mathsf{SharedEventRec}(W_t^A,W_t^B).
\]
The helper transcript $\mathcal H_t$ is public and is included in $V_{\rm pre}^{\rm free}$. Acceptance/rejection patterns, retries, timings, and other zero-cost pre-challenge leakage are included there as well. The intended correctness properties are
\begin{align}
\Prb[E_t^A=E_t^B\neq\bot] &\ge 1-\eps_{\rm rec},\\
\Prb[E_t^A\neq E_t^B,\ E_t^A,E_t^B\neq\bot] &\le \delta_{\rm sync},
\end{align}
with bounded honest work and memory. If reconciliation fails, the protocol does not create honest-side candidate branches. It pauses, aborts, or rolls back according to policy.

Accepted event encodings must be canonical and self-delimiting. With a prefix-free or length-delimited encoding $\operatorname{Enc}$,
\begin{equation}
U_t
=H_{\rm ctx}
(\texttt{CTX},k,a_k,t,U_{t-1},\operatorname{Enc}(E_t)).
\label{eq:context-update}
\end{equation}
Including $(k,a_k,t)$ prevents accidental replay, cross-epoch confusion, and ambiguous concatenation.

\subsection{Streaming readiness update}

For each accepted event and each lane, define
\begin{align}
X_{t,j}^{(0)}&=X_{t-1,j},\\
X_{t,j}^{(r)}
&=H_{\rm rd}
(\texttt{READY},k,a_k,t,j,r,U_t,X_{t,j}^{(r-1)}),
\quad r=1,\ldots,L,\\
X_{t,j}&=X_{t,j}^{(L)}.
\label{eq:streaming-readiness}
\end{align}
The honest state width devoted to readiness is
\begin{equation}
w=J\ell.
\label{eq:w-def}
\end{equation}
The nominal readiness work per accepted event is $JL$ evaluations of the readiness primitive.

\paragraph{Why streaming instead of an epoch-end gate.}
The sequential work is distributed across the period in which shared experience is accumulating. The honest device therefore does not need to perform a large burst of work after the epoch has ended. More importantly, the security mechanism aligns with the underlying intuition: contextual branches become progressively more expensive to reconstruct as they age.

\subsection{Freeze and fresh random-access challenge}

After $T$ accepted events, the pending context and readiness vector are frozen. Before any cheap digest of the final pending state is published, the verifier samples a fresh random lane
\[
j\overset{\$}{\leftarrow}\{1,\ldots,J\}
\]
and fresh nonce $r$. A directional response has the idealized form
\begin{equation}
R_B
=H_{\rm resp}
(\texttt{B2A},k,a_k,j,r,TH,X_{T,j}).
\label{eq:response}
\end{equation}
Here $TH$ is a canonical digest of the public epoch transcript and direction-specific session data. The reverse direction uses an independently sampled $(j',r')$ and a distinct domain string \texttt{A2B}. Reflection and session substitution are therefore excluded by construction. A deployment limits proof opportunities per direction and attempt; retries use a fresh attempt identifier.

The response deadline $\Delta$ is measured from the verifier's local challenge-send time. Adversarial packet delay cannot extend the verifier's deadline; it can only harm liveness.

The retained-state theorem targets reconstruction of the selected final readiness state $X_{T,j}$ by the deadline.  The nonce-bound response is a conventional authentication wrapper around that state.  If the chosen wrapper has forgery probability at most $\eps_{\rm auth}$ without the correct selected-lane state, then end-to-end one-shot impersonation is bounded by the reconstruction probability plus $\eps_{\rm auth}$.  Equation~\eqref{eq:response} is the random-oracle idealization of that wrapper; a concrete implementation may instead use a standard nonce-bound MAC/KDF construction.

\subsection{Transactional advance and rollback}

After both directional challenges succeed, the parties may derive a new confirmed state from the realized contextual and readiness state, for example
\begin{equation}
R_T^{\rm root}
=H_{\rm root}(\texttt{ROOT},X_{T,1},\ldots,X_{T,J}),
\end{equation}
\begin{equation}
S_{k+1}
=H_{\rm next}
(\texttt{NEXT},S_k,k,a_k,\sigma_k,U_T,R_T^{\rm root}).
\label{eq:next-state}
\end{equation}
A lossy network cannot guarantee an instantaneous simultaneous commit. The honest implementation may therefore retain the pair $(S_k,S_{k+1})$ during a bounded handover window. This is an $O(1)$ transactional duplication of the realized trajectory, not candidate enumeration. If handover fails, the protocol rolls back or resynchronizes according to a bounded policy.

\section{Threat Model and Security Objectives}

The adversary may observe all public traffic, helper data, timing information, acceptance/rejection patterns, aborts, retries, attempt identifiers, salts, challenges, and confirmation messages.  The source model may also allow event injection, delay, suppression, or partial observation.  Every such capability is reflected in the posterior over accepted shared experience; raw sensor entropy is never credited directly.

The adversary is unrestricted in algorithmic strategy but bounded by the resources that can affect the response before the deadline.  The reconstruction game strengthens the attacker by revealing the true mature branch at selection time and by assigning no security credit to secrecy of the pre-epoch state $S_k$ or other initialization data.  Security therefore comes from preparation induced by contextual uncertainty, not from hiding an old key.

Three online quantities are calibrated for a deadline $\Delta$:
\begin{itemize}[leftmargin=2em,itemsep=0.2em]
\item $d_\Delta$: maximum sequential dependency depth executable after selection;
\item $Q_\Delta$: maximum number of target-bearing candidate values that can be tested in one selected-target reconstruction;
\item $b_\Delta(A)$: residual table-dependent cross-target protected query/response transcript complexity across the counterfactual reconstructions used for a target set $A$.
\end{itemize}
The preselection quantity $M_\Delta$ is the log-cardinality of state that remains accessible quickly enough to influence the response.  The adversary may choose query addresses adaptively, use massive parallelism, query selected or unselected protected prefixes, and store any nonlinear encoding before selection.  A dependent round is a computation layer whose queries may depend on earlier rounds but not on answers returned in the same round.

The main theorem is quantitative: for a target success probability $p$, it computes a necessary lower bound $M_{\min}(p;d_\Delta,Q_\Delta,b_\Delta,\pi)$ on retained state.  Thus a deployment obtains a direct security test:
\begin{equation}
C_{\rm RA}(\Delta)<M_{\min}
\quad\Longrightarrow\quad
\Pr[\mathsf{reconstruct\ selected\ readiness\ state}]<p,
\label{eq:deployment-security-test}
\end{equation}
where $C_{\rm RA}(\Delta)$ is the attacker's deadline-reachable storage capacity under the deployment's systems model.  Security is parameterized directly by retained-state capacity: for each specified success level, the theorem gives the minimum deadline-reachable state compatible with that level under the calibrated online resource bounds.

For the outer response authenticator, let $\eps_{\rm auth}$ bound a successful nonce-bound forgery without the correct selected readiness state.  Then the usual composition gives
\begin{equation}
\Pr[\mathsf{impersonation}]
\le
\Pr[\mathsf{selected\ state\ reconstruction}]+\eps_{\rm auth}.
\label{eq:auth-composition}
\end{equation}
The remainder of the proof therefore focuses on the reconstruction term.

\section{Mature Contextual Coverage}

\subsection{Accepted-source posterior and coverage}
A deployment may start from noisy correlated observations.  Security is defined on the accepted event process after reconciliation, not on raw sensor entropy.  For the mature accepted prefix $\mathbf E_\tau$, define
\begin{equation}
\pi_v(e)=\Prb[\mathbf E_\tau=e\mid V_{\rm pre}^{\rm free}=v],
\label{eq:source-posterior}
\end{equation}
where the free preselection view contains helper data, accept/reject history, retries, timing, and all other zero-cost information available before challenge selection.  This convention automatically charges reconciliation leakage and selective acceptance to the source posterior.

For any discrete posterior $\pi$ over candidate histories, define the best mass covered by $q$ candidates and the corresponding coverage number by
\begin{align}
C_q(\pi)
&=\max_{A:\,|A|\le q}\sum_{u\in A}\pi(u),
\label{eq:coverage-q}\\
K_p(\pi)
&=\min\{q:C_q(\pi)\ge p\}.
\label{eq:Kp}
\end{align}
Coverage is the operational source quantity: it asks how many histories must be prepared to cover a desired success mass.  If the pointwise conditional min-entropy is at least $h$, then $C_q(\pi_v)\le q2^{-h}$ and
\begin{equation}
K_p(\pi_v)\ge\lceil p2^h\rceil.
\label{eq:Kp-minentropy}
\end{equation}
Exact posterior coverage is preferable whenever it can be estimated.

\subsection{Maturity frontier}
One accepted event adds $L$ sequential readiness layers.  If at most $d_\Delta$ dependent layers can be executed after selection, then the most recent
\begin{equation}
r_\Delta=\min\left\{T,\left\lfloor\frac{d_\Delta}{L}\right\rfloor\right\}
\label{eq:r-delta}
\end{equation}
events can be replayed inside the deadline.  The mature-prefix length is therefore
\begin{equation}
\boxed{
\tau_\Delta
=T-r_\Delta
=\max\left\{0,T-\left\lfloor\frac{d_\Delta}{L}\right\rfloor\right\}.}
\label{eq:tau-delta}
\end{equation}
Only uncertainty in $\mathbf E_\tau=(E_1,\ldots,E_{\tau_\Delta})$ receives one-epoch security credit.  Recent uncertainty is deliberately excluded because it can be regenerated after selection.

A convenient sufficient source condition is
\begin{equation}
\max_e\Prb[E_t=e\mid E_{<t}=e_{<t},V_{\rm pre}^{\rm free}=v]
\le2^{-\eta_t},
\qquad t\le\tau_\Delta.
\label{eq:mature-event-condition}
\end{equation}
The chain rule then gives
\begin{equation}
\Hinf(\mathbf E_\tau\mid V_{\rm pre}^{\rm free}=v)
\ge h_{\rm src}:=\sum_{t=1}^{\tau_\Delta}\eta_t,
\label{eq:hsrc}
\end{equation}
and if $\eta_t\ge\eta$,
\begin{equation}
h_{\rm src}
\ge
\eta\max\left\{0,T-\left\lfloor\frac{d_\Delta}{L}\right\rfloor\right\}.
\label{eq:mature-rate}
\end{equation}
A measured or certified coverage bound may be used directly instead of this min-entropy route.

\subsection{Compact folding and coverage transfer}
The contextual accumulator $U_t$ is a deterministic folding of the accepted history, so hashing cannot be assumed to preserve entropy automatically.  Choose a high-probability source core $\mathcal A_v$ with
\begin{equation}
\Prb[\mathbf E_\tau\in\mathcal A_v\mid v]
\ge1-\eps_{\rm tail},
\qquad
K_v:=|\mathcal A_v|,
\label{eq:source-core}
\end{equation}
for example the smallest top-probability set of mass $1-\eps_{\rm tail}$.

For two distinct source-induced histories in this core, a random-oracle union bound over at most $T$ aligned contextual updates gives
\begin{equation}
\Prb[\text{their contextual trajectories merge by time }T]
\le\frac{T}{2^{\lambda}},
\label{eq:pair-fold-collision}
\end{equation}
and hence
\begin{equation}
\delta_{\rm fold}
\le
\binom{K_v}{2}\frac{T}{2^{\lambda}}.
\label{eq:core-fold-collision}
\end{equation}
This bound applies when the source-induced event sequences are fixed independently of an adaptive search for $H_{\rm ctx}$ collisions.  If the source model permits active collision search using $q_{\rm ctx}$ oracle answers, the corresponding adaptive random-oracle collision term is added and $q_{\rm ctx}$ is part of the attacker model.

Let $B_\tau$ denote the folded mature branch.  On injectivity over the source core, $q$ folded branches contain at most $q$ histories from the core, while probability outside the core is at most $\eps_{\rm tail}$.  Therefore
\begin{equation}
\boxed{
C_q(B_\tau\mid v)
\le
C_q(\mathbf E_\tau\mid v)+\eps_{\rm tail},}
\label{eq:coverage-transfer}
\end{equation}
and for $p>\eps_{\rm tail}$,
\begin{equation}
K_p(B_\tau\mid v)
\ge
K_{p-\eps_{\rm tail}}(\mathbf E_\tau\mid v).
\label{eq:Kp-transfer}
\end{equation}
Under the min-entropy condition, this implies
\begin{align}
C_q(B_\tau\mid v)&\le q2^{-h_{\rm src}}+\eps_{\rm tail},
\label{eq:coverage-transfer-h}\\
K_p(B_\tau\mid v)&\ge\left\lceil(p-\eps_{\rm tail})2^{h_{\rm src}}\right\rceil.
\label{eq:Kp-transfer-h}
\end{align}
A convenient sizing rule for $\delta_{\rm fold}\le2^{-\kappa}$ is
\begin{equation}
\lambda
\gtrsim
2\log_2K_v+\log_2T+\kappa,
\label{eq:lambda-core-rule}
\end{equation}
up to the small $K_v(K_v-1)/2$ correction.  The high-probability core, rather than min-entropy alone, therefore controls contextual-state width.

\section{Causal Readiness and the Retained-State Bound}

The security mechanism can be summarized before the proof: if a non-negligible success probability requires readiness for $K$ mature branch--lane targets, then roughly $K\ell$ protected bits must cross the selection cut, reduced only by the explicitly counted post-selection candidate/depth choices and protected information actually obtained online.  The following construction makes that statement precise for arbitrary preselection encoding.

\subsection{Standard-ROM experiment and protected-slice coupling}

The security proof is stated in the standard random-oracle model.  A deferred-decision representation isolates the finite protected readiness values used in the counting argument; it does not change the oracle available to the adversary.  We treat
\[
H_{\rm ctx},\qquad H_{\rm rd}
\]
as independent random oracles for the core reconstruction theorem; a concrete implementation may use one hash with explicit domain separation.

Fix a preselection free view $V_{\rm pre}^{\rm free}=v$ of positive probability and a finite mature branch--lane target set $A$ of size $K$.  All target weights below are the posterior weights induced by this fixed $v$.  Assume $\tau_\Delta\ge1$; if $\tau_\Delta=0$, no contextual prefix is mature and this one-epoch source mechanism receives no coverage credit.  From the readiness state immediately before the update associated with event $\tau_\Delta$ to the final epoch state, the honest chain contains
\[
(T-\tau_\Delta+1)L=(r_\Delta+1)L>d_\Delta
\]
dependent readiness transitions.  For the lower-bound game we keep only the final $d_\Delta+1$ transitions of this path and reveal the state immediately before them as a free cold anchor.  Discarding earlier transitions and revealing this branch-specific anchor can only help the adversary, while leaving exactly the one-round causal deficit needed by the counting argument.

For target $i\in A$, write this final protected readiness suffix as
\begin{equation}
Y_i=(Y_{i,0},\ldots,Y_{i,d}),\qquad d=d_\Delta,
\label{eq:protected-suffix}
\end{equation}
with $Y_{i,-1}$ the cold anchor immediately before the suffix.  The lower-bound game gives the selected target's cold anchor to the online adversary for free.

Each protected transition has a readiness-oracle input of the form
\[
H_{\rm rd}(\rho_{i,s},Y_{i,s-1}),
\]
where $\rho_{i,s}$ contains every input field except the predecessor state: the domain string, epoch and attempt identifiers, event/layer indices, lane identifier, and contextual accumulator for the corresponding branch.  Define the random-function slice
\begin{equation}
F_{i,s}(x):=H_{\rm rd}(\rho_{i,s},x),
\qquad x\in\bits^\ell.
\label{eq:ro-slice}
\end{equation}

\begin{assumption}[Source--readiness separation]
Conditioned on $V_{\rm pre}^{\rm free}=v$, the credited target distribution, the finite source core, and the protected-prefix labels $\rho_{i,s}$ are measurable with respect to source randomness, $H_{\rm ctx}$, and public protocol data that are independent of the as-yet-unseen protected $H_{\rm rd}$ slices.  Equivalently, contextual security credit is not assigned to a source/core selection rule that is itself chosen as a function of protected readiness-oracle outputs.  Attacker-controlled source choices that violate this separation receive no contextual security credit.
\label{ass:source-readiness-sep}
\end{assumption}

Let $\mathsf{GoodCtx}_A$ be the event, determined entirely by the source view and $H_{\rm ctx}$, that the contextual accumulators associated with distinct mature branches in $A$ remain separated throughout the protected suffix.  Because $\rho_{i,s}$ also contains the lane and transition indices, $\mathsf{GoodCtx}_A$ makes all prefixes $\rho_{i,s}$ distinct.  Restrictions of a random oracle to these disjoint prefix domains are therefore independent random functions $F_{i,s}$.

This gives an exact hidden-point coupling of the standard ROM.  For every protected coordinate $(i,s)$, independently sample a full base function
\[
G_{i,s}:\bits^\ell\to\bits^\ell
\]
and an independent value $Y_{i,s}\leftarrow\bits^\ell$.  Let the hidden point be
\[
P_{i,0}=Y_{i,-1},\qquad
P_{i,s}=Y_{i,s-1}\quad(s\ge1),
\]
and define
\begin{equation}
F_{i,s}(x)=
\begin{cases}
Y_{i,s}, & x=P_{i,s},\\
G_{i,s}(x), & x\ne P_{i,s}.
\end{cases}
\label{eq:hidden-point-coupling}
\end{equation}
Since $P_{i,s}$ is determined before the independent slice $F_{i,s}$ is sampled, replacing the base-function value at that one point by an independent uniform value leaves $F_{i,s}$ distributed exactly as a uniform random function.  Recursing over $s$ therefore produces exactly the same joint distribution as the original standard-ROM readiness chain.  The coupling is a proof representation, not a change to the adversary's oracle.

Let
\begin{equation}
\mathcal Y_A:=\prod_{i\in A}\prod_{s=0}^{d}\bits^\ell,
\qquad
D_A:=K(d+1)\ell.
\label{eq:protected-table-space}
\end{equation}
The finite protected table $Y_A\in\mathcal Y_A$ is uniform conditioned on $(v,\mathsf{GoodCtx}_A)$ and on all base functions $G_{i,s}$ and random-oracle slices outside the protected prefixes.  Denote this fixed proof environment by $\Gamma_A=\gamma$.  Importantly, $\Gamma_A$ is independent of $Y_A$ in the coupling.  Wrong predecessor guesses are answered by the fixed base functions $G_{i,s}$; only a query at the hidden point $P_{i,s}$ reveals the protected value $Y_{i,s}$.

The preselection retained state is an arbitrary encoding
\begin{equation}
M=\mathsf{Enc}(v,\gamma,Y_A)\in\mathcal M,
\qquad
\log_2|\mathcal M|\le m.
\label{eq:advice-encoder}
\end{equation}
No linearity, per-branch layout, or separability of $\mathsf{Enc}$ is assumed.  In particular, XOR/parity and nonlinear cross-target encodings are included.  ``Retained information'' is used operationally for this log-cardinality bound on the deadline-accessible state; it is not Shannon mutual information.

For each selected target $i$, an online decoder receives $(v,M,i,Y_{i,-1})$ and ordinary adaptive access to the original random oracles represented by the coupling above.  It must output $Y_{i,d}$ by the deadline.  The decoder may query protected slices $F_{j,s}$ for any plausible target $j$, not only the selected target.  Its dependent depth is at most $d_\Delta$, and at most $Q_\Delta$ concrete values are submitted in target-bearing positions relevant to reconstructing the selected chain, including protected predecessor candidates and the final output candidate.

Cross-target protected queries can carry genuine online information about the protected table through table-dependent query addresses and, on correct hidden-point hits, through protected values returned online.  They are therefore accounted for explicitly.  For a deterministic reconstruction of selected target $i$, let $Z_i(Y_A)$ be the ordered subtranscript of every query to a protected slice $F_{j,s}$ with $j\ne i$ that is completed before the final output.  Each record contains the round/order metadata, the slice identifier $(j,s)$, the queried predecessor string, and the returned value.  Recording the full foreign protected subtranscript is conservative but important: a wrong-point response is fixed once its query address and the base function are fixed, yet the query address itself may still depend on the protected table.  For a target set $A$ and fixed retained-state value $m_0$, concatenate the canonical counterfactual subtranscripts as
\[
Z_A(Y_A;m_0,\gamma):=(Z_i(Y_A))_{i\in A}.
\]
Define the aggregate online protected side-information complexity
\begin{equation}
b_\Delta(A)
:=
\max_{m_0,\gamma}
\log_2\left|
\left\{Z_A(y;m_0,\gamma):
y\in\mathcal Y_A,
\ \mathsf{Enc}(v,\gamma,y)=m_0
\right\}
\right|.
\label{eq:online-sideinfo}
\end{equation}
The definition is a conditional transcript log-cardinality: it counts only the residual table dependence of cross-target protected access after the retained state is fixed.  A query address or returned value already determined by $M$ contributes no additional $b_\Delta$; newly exposed protected information contributes according to the number of transcript possibilities it creates.  If the foreign protected subtranscript is fixed within each retained-state fiber, then $b_\Delta(A)=0$.  In particular, this holds when no table-dependent foreign protected query occurs.  Recording the complete protected subtranscript is intentionally sufficient for deterministic replay and avoids treating either a table-dependent helper address or its response as free information.

The proof may condition on $\Gamma_A=\gamma$, but physical executions still pay for all computation and oracle evaluations.  Randomized adversaries are reduced to deterministic ones below.

\subsection{Readiness lead and the causal selection cut}

Number post-selection round boundaries by $r=0,1,\ldots,d$, where $d=d_\Delta$.  Queries issued in round $r+1$ may depend on values known at boundary $r$, and their answers become available at boundary $r+1$.  The final reconstruction output is due at boundary $d$.

If protected state $Y_{i,s}$ first becomes available for use at boundary $r$, define its readiness lead by
\begin{equation}
\Lambda(i,s,r)=s-r.
\label{eq:readiness-lead}
\end{equation}
The free cold anchor $Y_{i,-1}$ is known at boundary zero and therefore has lead $-1$.  A correct ordinary readiness transition maps a predecessor known at boundary $r$ to its successor at boundary $r+1$; increasing both suffix position and boundary index by one preserves $\Lambda$.  Cold forward computation therefore remains at lead $-1$ throughout the suffix.

If $Y_{i,s}$ is available at boundary $r$, the remaining path to $Y_{i,d}$ has depth $d-s$.  Deadline success requires
\[
r+(d-s)\le d,
\]
which is equivalent to $\Lambda\ge0$.  Thus every successful reconstruction must obtain at least one selected-chain value with nonnegative lead from something other than ordinary cold forward computation.

\subsection{Causal witnesses with cross-target online side information}

For the reconstruction execution of selected target $i$, let $r_i(s)$ be the earliest round boundary at which the exact protected value $Y_{i,s}$ is operationally available for subsequent use.  A value is operationally available at boundary $r$ if it has been returned by a query completed by that boundary, or if the decoder already holds the exact string and can name it in a target-bearing query issued in round $r+1$; for $s=d$, emitting it as the final output at boundary $r$ also counts.  Set $r_i(-1)=0$ for the free cold anchor.  A latent local string that never enters a protected predecessor position or the final output is irrelevant to reconstruction and is not counted as an exposure.

\begin{lemma}[First timely witness is supplied rather than derived]\label{lem:first-timely-witness}
If the decoder for target $i$ outputs $Y_{i,d}$ by boundary $d$, then there exists a suffix position $s_i\in\{0,\ldots,d\}$ such that
\[
r_i(s_i)\le s_i,
\]
and the first timely availability of $Y_{i,s_i}$ is a concrete target-bearing candidate exposure rather than the answer of the correct ordinary transition $F_{i,s_i}(Y_{i,s_i-1})$.
\end{lemma}

\begin{proof}
Success gives $r_i(d)\le d$, so at least one suffix position satisfies $r_i(s)\le s$.  Let $s_i$ be the smallest such position.  If $s_i=0$, the correct transition $F_{i,0}(Y_{i,-1})$ can return only at boundary one, so availability by boundary zero must come from a supplied candidate.  If $s_i>0$, minimality and integrality give $r_i(s_i-1)\ge s_i$.  A correct transition producing $Y_{i,s_i}$ can be issued only after the predecessor is available, and its answer appears one boundary later.  It would therefore satisfy
\[
r_i(s_i)\ge r_i(s_i-1)+1\ge s_i+1,
\]
contradicting $r_i(s_i)\le s_i$.  Hence the first timely selected-chain value is supplied from retained or online-derived information and is exposed in a target-bearing predecessor position, or, when $s_i=d$, as the final output candidate.
\end{proof}

For a successful target, call the position $s_i$ chosen above its \emph{witness position}, and call the corresponding concrete exposure of the exact string $Y_{i,s_i}$ its \emph{witness exposure}.  There are at most $Q_\Delta$ target-bearing candidate exposures in the reconstruction execution, where an exposure is an $\ell$-bit value submitted in any protected readiness-predecessor position relevant to the reconstruction, or emitted as the final selected-state candidate.  This includes predecessor values used in foreign protected queries.  Fix any deterministic ordering of these exposures by round and local operation number.

\begin{lemma}[Causal certificate counting with online protected side information]\label{lem:causal-certificate}
Fix a target set $A$ of size $K$, a retained-state value $m_0\in\mathcal M$, and a vector of witness positions $(s_i)_{i\in A}$.  Suppose every target in $A$ is reconstructed successfully by the corresponding adaptive decoder.  Then the number of protected tables in $\mathcal Y_A$ compatible with that retained-state value and witness-position vector is at most
\begin{equation}
2^{b_\Delta(A)}
\bigl(1+Q_\Delta\bigr)^K
2^{D_A-K\ell}.
\label{eq:certificate-fiber}
\end{equation}
No branchwise decomposition of the preselection state is assumed, and protected queries may address unselected targets.
\end{lemma}

\begin{proof}
Fix the retained-state value and additionally fix one aggregate cross-target protected subtranscript $Z_A=z$.  By Eq.~\eqref{eq:online-sideinfo}, at most $2^{b_\Delta(A)}$ such transcripts need be considered.

For each successful table, delete the $K$ selected-chain witness values $Y_{i,s_i}$ and keep every other protected coordinate.  Record, for every target $i$, the ordinal of its witness exposure among these target-bearing candidate exposures; each ordinal has at most $1+Q_\Delta$ possibilities.  If a deleted witness occurs literally as a queried predecessor or returned value in $z$, it is already recovered from the recorded subtranscript.

Replay the deterministic reconstruction for target $i$ from $(v,m_0,i,Y_{i,-1})$.  The fixed subtranscript $z$ records every foreign protected query and response, so each such interaction is reproduced with its recorded address and answer.  Non-protected oracle behavior and wrong-point answers on the selected protected slices are determined by the fixed proof environment $\Gamma_A$ once their replayed addresses are fixed.  Thus adaptive helper behavior is replayed without treating either a foreign protected address or its response as free; every residual table-dependent variation of those records has already been charged by the factor $2^{b_\Delta(A)}$.

For any witness not already present in $z$, an induction over replay operations shows that the execution reaches its recorded exposure without the deleted value.  The initial replay state is fixed by $(v,m_0,i,Y_{i,-1})$ and the undeleted protected coordinates.  Foreign protected interactions are fixed by $z$, and all other oracle answers before the witness exposure are fixed by $\Gamma_A$ and the replayed query addresses.  The only remaining way for the exact deleted witness to enter the execution earlier would be a correct selected-chain transition returning it, which would make it operationally available by ordinary forward computation and contradict Lemma~\ref{lem:first-timely-witness}.  A query that submits the deleted value as a protected predecessor is itself a target-bearing exposure, so its first such occurrence is no later than the recorded ordinal.  The concrete candidate at that exposure reconstructs $Y_{i,s_i}$.  Repeating this reconstruction for all $i\in A$ recovers all deleted witnesses.

For fixed $(m_0,z,(s_i))$, the encoding consisting of the undeleted $D_A-K\ell$ protected bits and the $K$ candidate ordinals is therefore injective.  Multiplying by the at most $2^{b_\Delta(A)}$ possible foreign protected subtranscripts proves Eq.~\eqref{eq:certificate-fiber}.
\end{proof}

Thus cross-target helpers are neither prohibited nor treated as free: all residual protected-table dependence in their addresses and responses is charged through $b_\Delta(A)$.  Appendix~\ref{app:robustness} gives concrete two-value examples showing why this transcript term is needed.

\subsection{Perfect causal readiness}

Define the causal width
\begin{equation}
\boxed{
a_{\rm causal}
:=
\pos{
\ell
-\log_2\!\bigl(1+Q_\Delta\bigr)
-\log_2(d_\Delta+1)}.}
\label{eq:a-causal}
\end{equation}

\begin{theorem}[Perfect causal readiness with online side information]\label{thm:perfect-causal}
Fix $V_{\rm pre}^{\rm free}=v$, a $K$-target set $A$ satisfying $\mathsf{GoodCtx}_A$, and an arbitrary proof environment $\Gamma_A=\gamma$.  Suppose the adversary's deadline-accessible preselection state has at most $2^m$ possible values and every target in $A$ is answered successfully for every protected table $Y_A\in\mathcal Y_A$ by a decoder of dependent depth at most $d_\Delta$ and candidate budget at most $Q_\Delta$.  Let $b_\Delta(A)$ be defined by Eq.~\eqref{eq:online-sideinfo}.  Then
\begin{equation}
\boxed{m\ge K a_{\rm causal}-b_\Delta(A).}
\label{eq:perfect-causal-bound}
\end{equation}
\end{theorem}

\begin{proof}
For one retained-state value there are at most $(d_\Delta+1)^K$ witness-position vectors.  Lemma~\ref{lem:causal-certificate} therefore bounds its total successful-table fiber by
\[
2^{b_\Delta(A)}
(d_\Delta+1)^K
(1+Q_\Delta)^K
2^{D_A-K\ell}.
\]
There are at most $2^m$ retained-state values.  Perfect readiness requires these fibers to cover all $2^{D_A}$ protected tables.  Taking logarithms gives Eq.~\eqref{eq:perfect-causal-bound}, clipped at zero.\end{proof}

\subsection{Partial success from exact mature coverage}

Let a finite mature target core be $\mathcal I=\{1,\ldots,N\}$ with original absolute selection weights $w_i$ and total core mass
\begin{equation}
W_c=\sum_{i\in\mathcal I}w_i\le1.
\end{equation}
A contextual branch of posterior probability $\pi_b$ contributes $J$ target weights $\pi_b/J$ under a uniform lane challenge.

For $0<\beta\le W_c$, define
\begin{equation}
K_\beta(w)
=
\min\left\{|A|:
A\subseteq\mathcal I,
\sum_{i\in A}w_i\ge\beta
\right\}.
\label{eq:Kbeta-target}
\end{equation}
Let $\mathsf{GoodCtx}_{\mathcal I}$ denote contextual separation for the entire finite core; then every subset $A\subseteq\mathcal I$ inherits $\mathsf{GoodCtx}_A$.  All probabilities in this subsection are pointwise in the already fixed free view $v$.  In particular, let
\begin{equation}
p=p(v):=\Prb[\mathsf{Succ}\mid V_{\rm pre}^{\rm free}=v]
\label{eq:pointwise-success}
\end{equation}
be the conditional one-epoch probability of reconstructing the selected final readiness state by the deadline.  Collect the explicit source/synchronization losses, also conditioned on $v$, in the union-bound quantity
\begin{equation}
\eps_{\rm bad}
:=
\min\left\{1,\eps_{\rm sync}+\eps_{\rm fold}+\eps_{\rm src}\right\},
\label{eq:eps-bad-decomp}
\end{equation}
where $\eps_{\rm sync}$ is an epoch-level synchronization-failure bound, $\eps_{\rm fold}$ bounds failure of the contextual-separation event (for example by Eq.~\eqref{eq:core-fold-collision}), and $\eps_{\rm src}$ is an optional failure probability attached to a statistically certified source model or confidence set. Thus the conditional probability of leaving the source-, synchronization-, or contextual-separation-good experiment is at most $\eps_{\rm bad}$. In the ideal source model $\eps_{\rm src}=0$.  The success mass that must be supplied inside the finite core is
\begin{equation}
p_c
=
pos{p-\eps_{\rm bad}-(1-W_c)}.
\label{eq:p-core}
\end{equation}

\begin{lemma}[Determinization]\label{lem:determinization}
Under the pointwise-view convention above, it is sufficient to analyze a deterministic encoder and deterministic online decoders.
\end{lemma}
\begin{proof}
Condition on the adversary's complete private random tape $R$.  The retained-state, dependent-depth, and query bounds hold pointwise for every fixing of $R$, while conditional success is the average of the corresponding deterministic success probabilities.  Hence some fixing of $R$ attains at least the average success without increasing any resource bound.  The counting argument may therefore take the encoder and all online decoders to be deterministic.\end{proof}

For $K\le N$, define the worst-case aggregate online side-information budget over $K$-target subsets by
\begin{equation}
\bar b_\Delta(K):=\max_{A\subseteq\mathcal I,\ |A|=K} b_\Delta(A).
\label{eq:bbar}
\end{equation}

\begin{theorem}[Partial causal readiness from exact coverage]\label{thm:partial-causal}
Fix $V_{\rm pre}^{\rm free}=v$.  Under the standard-ROM causal experiment above, for every $\beta\in(0,p_c)$ let $K=K_\beta(w)$.  Then
\begin{equation}
\boxed{
M_\Delta
\ge
K a_{\rm causal}
-\bar b_\Delta(K)
-\log_2\binom{N}{K}
-\log_2\frac{W_c-\beta}{p_c-\beta}.}
\label{eq:partial-causal-bound}
\end{equation}
The supremum over $\beta\in(0,p_c)$ is valid, with negative values replaced by zero.  The theorem uses the causal width $a_{\rm causal}$ of Eq.~\eqref{eq:a-causal}.
\end{theorem}

\begin{proof}
Apply Lemma~\ref{lem:determinization}.  It is convenient to prove the equivalent success upper bound first.  Work on the source-, synchronization-, and contextual-separation-good experiment for the finite core $\mathcal I$; Eq.~\eqref{eq:p-core} accounts for leaving this experiment and for selection outside the core.

Fix an arbitrary proof environment $\Gamma_{\mathcal I}=\gamma$.  Let $\mathcal Y_{\mathcal I}$ be the full protected table for all $N$ core targets, with $D_{\mathcal I}=N(d+1)\ell$, and for $y\in\mathcal Y_{\mathcal I}$ let $S_\gamma(y)$ be the total core selection mass of targets on which the deterministic adversary succeeds.  Define
\[
\mu_\beta(\gamma)
:=
\Pr_{Y_{\mathcal I}}[S_\gamma(Y_{\mathcal I})\ge\beta\mid\Gamma_{\mathcal I}=\gamma].
\]
Every table counted by $\mu_\beta(\gamma)$ has at least $K=K_\beta(w)$ successful core targets.  Choose a canonical $K$-subset $A$ of successful targets; there are at most $\binom NK$ choices.

Fix one such subset $A$ and one assignment of every protected coordinate outside $A$.  With that assignment fixed, the retained-state encoder becomes an arbitrary function of the remaining $K$-target table $Y_A$ (and of the already fixed proof environment); no separability is introduced by this conditioning.  Lemma~\ref{lem:causal-certificate}, followed by the sum over witness-position vectors, shows that one retained-state value can be successful on at most
\[
2^{D_A-Ka_{\rm causal}+b_\Delta(A)}
\]
$K$-target tables.  Summing over the $2^{D_{\mathcal I}-D_A}$ outside assignments, at most $2^{M_\Delta}$ retained-state values, and at most $\binom NK$ canonical subsets gives, for every fixed $\gamma$,
\begin{equation}
\mu_\beta(\gamma)
\le
B_\beta
:=
\min\left\{1,
2^{M_\Delta+\bar b_\Delta(K)+\log_2\binom NK-Ka_{\rm causal}}
\right\}.
\label{eq:mu-pointwise-gamma}
\end{equation}
The bound holds pointwise for every fixed $\gamma$.

Since $0\le S_\gamma\le W_c$,
\[
\mathbb E[S_\gamma(Y_{\mathcal I})\mid\gamma]
\le
\beta+(W_c-\beta)\mu_\beta(\gamma)
\le
\beta+(W_c-\beta)B_\beta.
\]
The right-hand side is independent of $\gamma$, so averaging over the proof environment preserves the same upper bound.  Restoring the bad-event and outside-core probabilities therefore yields
\begin{equation}
p
\le
\eps_{\rm bad}+(1-W_c)
+\beta+(W_c-\beta)B_\beta.
\label{eq:success-upper-proof}
\end{equation}
Equivalently, Eq.~\eqref{eq:p-core} gives
\[
p_c\le\beta+(W_c-\beta)B_\beta.
\]
For $0<\beta<p_c$, this implies
\[
B_\beta\ge\frac{p_c-\beta}{W_c-\beta}.
\]
Combining this with Eq.~\eqref{eq:mu-pointwise-gamma} and taking logarithms gives Eq.~\eqref{eq:partial-causal-bound}.\end{proof}

It is often more natural to quote the theorem in the forward security direction.  With
\begin{equation}
B_\beta
:=
\min\left\{1,
2^{M_\Delta+\bar b_\Delta(K)+\log_2\binom NK-Ka_{\rm causal}}
\right\},
\label{eq:good-table-upper}
\end{equation}
the proof gives directly
\begin{equation}
p
\le
\eps_{\rm bad}+(1-W_c)+\beta
+(W_c-\beta)B_\beta.
\label{eq:success-upper}
\end{equation}
Equations~\eqref{eq:good-table-upper}--\eqref{eq:success-upper} expose directly how retained-state length, finite target coverage, and online resources trade against success.

\subsection{Uniform-core and min-entropy corollaries}

For a uniform core of $B$ contextual branches and $J$ lanes, $N=JB$, $W_c=1$, and
\[
K_\beta=\lceil\beta N\rceil.
\]
For large $N$ and fixed $\beta$, the theorem with $\bar b_\Delta(\lceil\beta N\rceil)=0$ gives
\begin{equation}
\frac{M_\Delta}{N}
\gtrsim
\beta a_{\rm causal}-h_2(\beta)
-\frac{\bar b_\Delta(\lceil\beta N\rceil)}{N},
\label{eq:uniform-asymptotic}
\end{equation}
up to vanishing per-target terms and the good-table-mass penalty.  A positive linear-in-$N$ retained-state bound follows whenever a fixed $\beta>0$ in the admissible success range satisfies
\begin{equation}
\beta a_{\rm causal}>h_2(\beta).
\label{eq:uniform-positive-regime}
\end{equation}
To express the scaling as a constant fraction of the raw protected width $N\ell$, one additionally needs $a_{\rm causal}/\ell$ bounded below by a positive constant along the parameter family (or, equivalently for an $h\to\infty$ family, one may hold $\ell,Q_\Delta,d_\Delta$ fixed in a regime satisfying Eq.~\eqref{eq:uniform-positive-regime}).  Under that additional condition, $B=2^h$ and $w=J\ell$ give $M_\Delta=\Omega(w2^h)$ for finite-core families that continue to satisfy the source-separation and contextual-fold conditions.  More generally, a linear bound remains whenever the online side-information term grows strictly slower than the positive linear contribution $N[\beta a_{\rm causal}-h_2(\beta)]$.

For a nonuniform source, conditional min-entropy alone is not the preferred route to a numerical retained-state bound.  If a finite source core has $K_{\rm core}$ branches and point probabilities at most $2^{-h}$, then each branch--lane target has weight at most $2^{-h}/J$, so
\[
K_\beta(w)\ge\lceil\beta J2^h\rceil.
\]
The finite core size $N=JK_{\rm core}$ still enters the subset-selection term, while the causal width is set by $(\ell,Q_\Delta,d_\Delta)$ and the cross-target online contribution enters through $\bar b_\Delta(K)$.  Exact mature posterior coverage is therefore both tighter and conceptually primary.

\subsection{Interpretation}

The one-epoch theorem separates four effects.  Mature contextual coverage determines how many branch--lane targets must be supportable on a non-negligible fraction of protected tables.  The causal width measures the per-target contribution remaining after witness-position choice and online candidate exposure.  The term $\bar b_\Delta(K)$ charges table-dependent protected side information acquired from other targets after selection.  The finite-core combinatorial term accounts for success concentrating on different target subsets for different protected tables.

The bound counts information rather than layout: the retained state may be globally encoded across targets, while protected information newly acquired online is charged separately.  Appendix~\ref{app:physical} maps this abstract deadline-accessible state to physical memory tiers.

\section{Honest Cost and Streaming Feasibility}\label{sec:honest-cost}

The device-side cost does not grow with the number of plausible contextual histories.  The honest implementation has the following fixed surface:
\begin{center}
\begin{tabular}{p{0.24\linewidth}p{0.66\linewidth}}
\toprule
Operation & Honest-device action \\
\midrule
Persistent epoch state & one contextual accumulator, $J$ readiness lanes, bounded reconciliation state, and $O(1)$ confirmed states; \\
Accepted event & one contextual hash and $JL$ fixed-width readiness evaluations; \\
End of epoch & freeze the already-maintained state; no branch enumeration or catch-up computation; \\
Challenge & read one selected lane and compute one nonce-bound authenticator. \\
\bottomrule
\end{tabular}
\end{center}
Accordingly,
\begin{equation}
M_H
=J\ell+O(\lambda)+M_{\rm rec}
\label{eq:honest-memory}
\end{equation}
and the readiness work per accepted event is
\begin{equation}
W_H=JL.
\label{eq:honest-work-event}
\end{equation}
With $J=256$ and $\ell=256$, the readiness vector is 8~KiB; with $L=4$, one accepted event performs 1024 readiness evaluations.  Both $J$ and $L$ are tunable rather than protocol constants: for example, $J=64$, $\ell=256$, and $L=4$ use a 2~KiB readiness vector and 256 readiness evaluations per accepted event, with the corresponding security tradeoff evaluated by the theorem.  Platform suitability is therefore set by the chosen profile, the accepted-event rate, and measured hash throughput, not by the number of histories in the attacker's posterior.

A simple feasibility envelope illustrates the interaction between honest throughput and adversarial post-selection speed. Let the honest readiness-update rate be $f_H$, let events arrive at rate $\nu$, and let the adversary's maximum sequential rate be $f_A=s f_H$. Honest real-time processing requires
\[
JL\nu<f_H.
\]
To make the entire epoch deeper than the post-selection adversarial depth $f_A\Delta$, one needs
\[
TL>f_A\Delta.
\]
For some $L$ to satisfy both inequalities, a necessary envelope is
\begin{equation}
\boxed{
\tau_{\rm epoch}=\frac{T}{\nu}
>Js\Delta.}
\label{eq:streaming-envelope}
\end{equation}
This is an engineering relation, not a hardware benchmark. Streaming matters because the honest work is spread across $\tau_{\rm epoch}$ instead of being paid in a burst after the epoch.

\section{Cost, Calibration, and Illustrative Results}
Security calibration is separated from the fixed honest-device path of Section~\ref{sec:honest-cost}.  The attacker-side depth and candidate budgets must be measured against the readiness primitive rather than inferred from clock rate or generic hash throughput. Define
\begin{equation}
d_\Delta=\max\{d:T_{\rm chain}(d)\le\Delta\},
\qquad
Q_\Delta=\max\{q:T_{\rm batch}(q)\le\Delta\},
\label{eq:calibration}
\end{equation}
where $T_{\rm chain}$ measures a dependent readiness chain and $T_{\rm batch}$ measures independent target-bearing candidate evaluations under an attacker-optimized implementation. Prefix precomputation and device-specific batching are allowed.  A security profile may use attacker-favorable maxima for $d_\Delta$ and $Q_\Delta$, even when obtained on different platforms, or use a measured joint feasible region when tighter calibration is needed.  The cross-target term $\bar b_\Delta(K)$ is likewise an attacker-side analysis bound on the variability of successful protected helper transcripts; it is not a field or data structure maintained by the honest device.  A certified deployment therefore combines a simple honest profile $(J,L,T,\Delta)$ with attacker-favorable bounds on $(d_\Delta,Q_\Delta,\bar b_\Delta)$ and a source posterior model.

The theorem lower-bounds the log-cardinality of a \emph{deadline-accessible retained state}, not intrinsically RAM or Shannon information. If $C_{\rm RA}(\Delta)$ denotes the information capacity that can be randomly accessed quickly enough to affect a response within $\Delta$, then successful attack requires
\begin{equation}
C_{\rm RA}(\Delta)\ge M_\Delta.
\label{eq:physical-corollary}
\end{equation}
Calling the result a RAM lower bound additionally requires a systems argument that slower tiers cannot contribute within the deadline.

For illustration, with a uniform mature core of $2^{20}$ branches, $J=256$, $\ell=256$, $Q_\Delta=2^{40}$, $d_\Delta=32$, target success $p=0.9$, and $\bar b_\Delta(K)=0$, the bound gives approximately $M_\Delta\gtrsim5.92$ GiB.  An exact-coverage binary mature-source example with the same readiness parameters gives approximately $41.52$ GiB.  A nonzero online side-information term subtracts directly from these retained-state bounds.

Figure~\ref{fig:causal-h-main} shows the central scaling effect for uniform mature cores.  The horizontal axis is the number of uncertain mature branches on a logarithmic scale, while each line fixes the honest readiness width.  Honest memory remains fixed at the line's stated KiB value; the adversarial retained-state lower bound rises rapidly as more mature branches must be covered.  Line style and marker shape, rather than color, identify the three honest-memory profiles so the plot remains readable in grayscale.

\begin{figure}[H]
\centering
\includegraphics[width=0.70\linewidth]{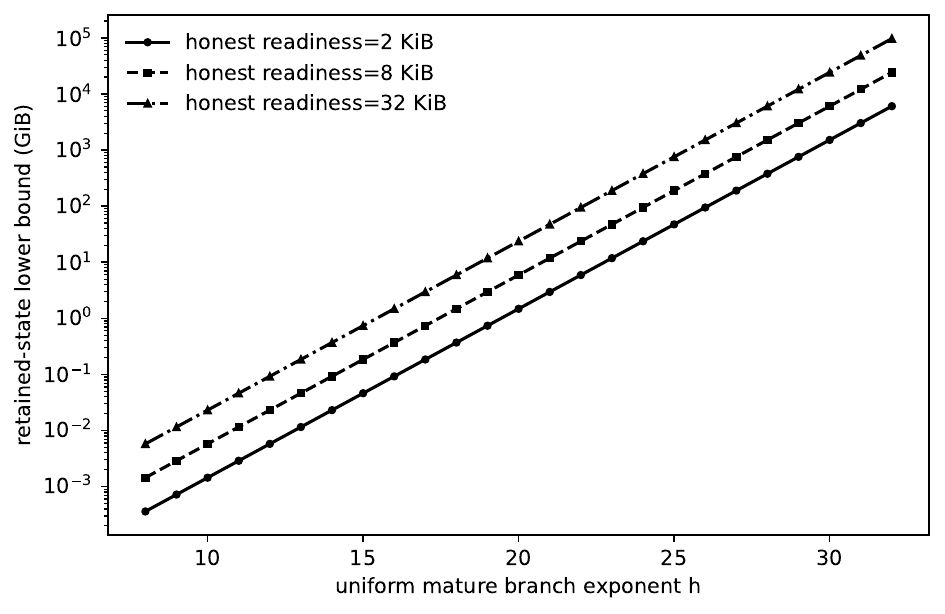}
\caption{Retained-state lower bound versus the number $2^h$ of uniform mature branches for three fixed honest readiness widths.  Solid-circle, dashed-square, and dash-dot-triangle curves correspond to 2, 8, and 32~KiB of honest readiness state.  Parameters are $p=0.9$, $\ell=256$, $Q_\Delta=2^{40}$, $d_\Delta=32$, and $\bar b_\Delta(K)=0$.}
\label{fig:causal-h-main}
\end{figure}

Full numerical details and reproducibility files are given in Appendix~\ref{app:numerics}.

\section{Related Work and Positioning}

\subsection{Post-quantum and resource-bounded cryptography}
Post-quantum cryptography addresses the possibility that quantum algorithms change the computational hardness assumptions underlying deployed public-key systems~\cite{Shor1994,NISTPQC2024}. Contextual Chain is complementary to post-quantum primitives: it adds deadline-accessible retained-state capacity as a separate resource dimension that can remain relevant even when computational throughput is high.

The closest conceptual precedents are resource-bounded cryptographic models. Maurer explicitly considered security based on restricted adversarial memory rather than restricted computation~\cite{Maurer1992}. The bounded-quantum-storage model develops analogous ideas when the scarce resource is coherent quantum storage~\cite{DamgardEtAl2008}, and bounded-storage authentication gives a modern authentication formulation~\cite{DodisQuachWichs2022}. Proofs of space make storage itself verifiable~\cite{DziembowskiEtAl2015}. Memory-hard functions instead design computations whose efficient evaluation requires substantial memory or sustained space~\cite{AlwenSerbinenko2015,AlwenBlockiPietrzak2018,ForlerLucksWenzel2014,BonehCorriganGibbsSchechter2016,RFC9106}. Contextual Chain differs in the source of the resource pressure: posterior uncertainty over a relationship-specific history creates many plausible readiness targets, and a fresh challenge selects one only after preparation. The theorem lower-bounds retained-state length needed across those targets rather than assuming a fixed storage cap or importing a password-hashing game.

A second, technically closer line of work studies preprocessing or non-uniform attacks. In the auxiliary-input random-oracle model, an attacker may retain an arbitrary bounded-length function of the random-oracle table and then make a bounded number of online oracle queries~\cite{CorettiDodisGuoSteinberger2018}. Subsequent work studies advice-size/online-query and offline/online-query tradeoffs for preprocessing attacks~\cite{GhoshalTessaro2023}. Static-memory-hard functions likewise distinguish preprocessed static storage from dynamic runtime memory~\cite{DryjaLiuPark2018}. These models are important precedents for the proof language used here: $M_\Delta$ is an advice/state-length quantity, not a sum of per-target mutual informations. The model likewise permits arbitrary bounded-length preprocessing of the protected readiness-oracle environment after conditioning on the source/$H_{\rm ctx}$ view, together with adaptive online oracle access.  Cross-target protected helper subtranscripts acquired online are represented explicitly by $b_\Delta$, so preprocessing advice and online side information are not conflated. The distinction is instead in how the protected task family is induced: after conditioning on public source leakage, the accepted-source posterior and maturation rule determine a relationship- and epoch-specific collection of readiness chains, and a fresh post-selection challenge chooses which member must be reconstructed under the deadline. The distinguishing contribution is the source-to-preparation transformation and its causal counting bound: accepted contextual uncertainty generates an epoch-specific family of readiness targets, and the challenge is chosen only after preparation.

\subsection{Contextual Chain design lineage}
Contextual Chain was first studied as an operational framework in which evolving context guides continuation, with an initial systems realization for post-partition recovery under intermittent connectivity~\cite{KimContextualChainSystems2026}.  Version~1 of the public specification records that recovery-oriented realization and treats resource-sensitive proof of context as a non-normative extension~\cite{KimCCSpecV1}.  Version~2 separates the general Contextual Chain design family from implementation-specific profiles: evolving context and context-informed continuation form the conceptual core, whereas synchronization, storage policy, and cryptographic authentication are realization-specific modules~\cite{KimCCSpecV2}.  It also identifies contextual resource asymmetry as a research direction without fixing an attacker model or claiming a lower bound.  The continuity-authentication construction studied here instantiates this design space with streaming readiness and a post-selection challenge, while the retained-state theorem formalizes the resulting preparation requirement.

\subsection{Reconciliation and noisy context}
Secure sketches and fuzzy extractors provide established tools for reconciling noisy, nonuniform observations while accounting for public helper information~\cite{DodisReyzinSmith2004}. Contextual Chain does not require the contextual source itself to be extracted into a uniform key. Instead, the security theorem takes the posterior over accepted mature histories, after all zero-cost helper information and selection leakage have been conditioned on, as its source interface. Reconciliation is therefore a replaceable source-layer primitive.

\subsection{Context-based authentication and co-presence}
Sensor context has previously been used to establish co-presence, resist relay attacks, or derive authentication material for IoT and zero-interaction authentication~\cite{TruongEtAl2015,MiettinenEtAl2018}.  Those systems demonstrate both the practical value of ambient context and the importance of modeling context manipulation and helper leakage.  Contextual Chain uses context differently: accepted shared experience is not required to become a uniform secret key or a one-shot co-presence classifier.  Instead, its residual posterior uncertainty determines how many mature readiness trajectories an uncertain observer must prepare for before a fresh post-selection challenge.  This distinction is why the source interface in Section~4 is posterior coverage rather than raw sensor entropy.

\subsection{Simple hash chains}
A deterministic chain such as
\[
U_{t+1}=\operatorname{SHA256}(\texttt{CTX}\concat U_t\concat E_{t+1})
\]
can bind a compact state to an event history, but it does not create entropy and by itself gives no retained-state lower bound. The security mechanism studied here additionally requires mature contextual uncertainty, streaming readiness, and a fresh post-selection challenge.

\section{Discussion}
Contextual Chain is intended for continuity after association, not for first-contact identity.  Its natural operating regime is therefore an intermittently connected relationship in which the peers can repeatedly obtain some shared physical or operational evidence but cannot assume a continuously available link.  Link loss does not create honest-side candidate histories: the state advances only on canonical accepted events, and failed reconciliation invokes a bounded pause, abort, rollback, or resynchronization policy.  The cryptographic source is consequently the posterior over \emph{accepted} mature histories after helper data, retries, timing, and allowed active manipulation have been conditioned on.  Appendix~\ref{app:source-details} makes this source interface explicit and explains why selective jamming and repeated reconciliation must be reflected in that posterior rather than treated as free entropy.

The security statement is quantitative rather than absolute.  A deployment chooses a response deadline $\Delta$ and then gives the attacker favorable bounds on post-selection sequential depth and candidate-query capacity.  The theorem converts these bounds and mature posterior coverage into a minimum deadline-accessible retained-state requirement.  The mapping from that abstract state length to RAM, cache, VRAM, or other storage depends on which tiers can actually affect the response before $\Delta$; Appendix~\ref{app:physical} gives this systems interpretation.  This distinction is important for the intended ``memory explosion'' view of security: the theorem does not assume that a named memory technology is scarce, but instead states how much state must already be reachable quickly enough to matter.

Cross-target online side information is separated from retained state rather than silently treated as free.  The adversary may query protected prefixes associated with other plausible histories, and the term $\bar b_\Delta(K)$ measures the residual protected-table information such accesses can reveal after the retained state is fixed.  A purely memory-only numerical claim therefore requires an upper bound on this term; $\bar b_\Delta(K)=0$ means that the foreign protected subtranscript is fixed after the retained state and proof environment are fixed; it does not prohibit foreign queries.  Appendix~\ref{app:robustness} explains why XOR/parity and re-queryable helper examples make this accounting necessary, while Appendix~\ref{app:calibration} describes how $d_\Delta$, $Q_\Delta$, and the online side-information term enter a deployment profile.  The theorem itself remains strategy-independent within those resource bounds.

The honest side stays deliberately simpler than the proof.  It maintains one context value and $J$ fixed-width readiness lanes, performs $JL$ readiness evaluations per accepted event, and reads only one lane at challenge time.  The parameters used in the numerical examples are not hardware benchmarks; on a constrained microcontroller the dominant implementation question may be energy or per-event hash throughput rather than state width.  This is why the sequential and parallel attack rates are calibration inputs rather than generic hash-speed constants.  Appendix~\ref{app:calibration} gives the corresponding measurement rule, and Appendix~\ref{app:numerics} shows how the resulting values propagate into the retained-state bound.  A small device benchmark is therefore the most useful next systems experiment, but it is not needed for the abstract security theorem.

Frequent disconnection also raises a longer-timescale question that one epoch cannot settle.  Missing a deadline does not prevent an attacker from continuing to reconstruct old context later.  The basic protocol handles ordinary interruption with bounded transactional state, whereas a stronger persistent profile can carry the wide readiness vector across epochs so that unresolved preparation debt is not compressed into a short state once the attacker eventually catches up.  Appendix~\ref{app:persistent} formulates this maturation-versus-catch-up extension.  The same separation clarifies temporary snapshot compromise: once a leaked state is treated as free information, later context can create a new preparation asymmetry if fresh residual uncertainty again matures.  Appendix~\ref{app:healing} gives this contextual-healing formulation and also states the unavoidable boundary that a perfect clone receiving the same future state and context cannot be distinguished by continuity alone.

The broader point is that Contextual Chain does not replace conventional authentication or post-quantum cryptography.  It adds a different axis of asymmetry.  Computational hardness asks how much computation is needed after the challenge; Contextual Chain asks how much useful state must already have been prepared before a challenge reveals what matters.  Memory-based and preprocessing-based security are established ideas, but here the preparation target is generated continuously by relationship-specific shared experience.  That source-to-preparation transformation is what allows a very small streaming endpoint to induce a potentially much larger adversarial readiness problem without making the endpoint itself enumerate the uncertainty.

\section{Conclusion}
Contextual Chain is built around a deliberately simple asymmetry.  The legitimate device follows one realized shared history, keeps one compact context state and a fixed readiness vector, and answers from one randomly selected lane.  It never enumerates alternatives.  An attacker that has lost contextual continuity may instead need to remain prepared across many mature histories before learning which lane and history must be answered under the deadline.

The causal counting theorem turns that asymmetry into a quantitative security statement.  It allows arbitrary nonlinear preselection encoding and adaptive cross-target oracle use, then lower-bounds the deadline-accessible state required for a chosen success probability after accounting for post-selection depth, candidate queries, and protected information newly obtained online.  Exact mature posterior coverage connects the source layer to that bound, while a standard nonce-bound authenticator composes the reconstructed lane state into an end-to-end response.

The resulting separation is practical: protocol execution remains a fixed-memory stream of hash-based updates, whereas security is tuned through source coverage and measured attacker resources.  This makes the construction suitable for evaluation on lightweight and IoT-class devices without making the device implementation inherit the combinatorial complexity faced by an uncertain attacker.

\appendix
\section{Proof Robustness and Deferred-Decision Details}\label{app:robustness}
The retained-state theorem is a global compression statement rather than a sum of targetwise information bounds.  This appendix records the proof choices that matter most when arbitrary preprocessing and adaptive cross-target queries are allowed.

\subsection{Why global counting is necessary}
A single resource-hardness parameter can hide distinct effects: source identification, post-selection sequential depth, candidate search, and information retained across the selection cut.  The main theorem therefore exposes mature posterior coverage, $d_\Delta$, $Q_\Delta$, and retained-state length separately.  Likewise, a multi-target lower bound cannot in general be obtained by adding marginal mutual informations.  Global encodings such as XOR/parity or nonlinear shared-root codes can carry little marginal information about each target while becoming useful after other values are revealed.  The proof instead counts fibers of an arbitrary retained-state encoder over the complete protected table.

Re-queryable oracle side information gives the same warning in a simpler form.  Let $Y$ be a protected $\ell$-bit value and let $Z=H_{\rm rd}(\texttt{PAD})$ be an independent, re-queryable non-protected value.  Retaining $M=Y\oplus Z$ gives $I(Y;M)=0$, yet re-querying $Z$ after selection recovers $Y$.  The retained state still contains $\ell$ bits, so this is not an attack on the theorem; it shows why a proof based only on mutual information would use the wrong quantity.

\subsection{Causal cut and cross-target protected side information}
Protected readiness addresses contain protected predecessor values, so a static list of addresses is not a suitable proof abstraction.  The readiness-lead invariant $\Lambda=s-r$ gives the causal boundary: cold selected-chain computation begins at lead $-1$, ordinary readiness transitions preserve lead, and deadline success requires a value with nonnegative lead.  Hence every successful target has a first timely witness that cannot first have been obtained by ordinary forward computation on the selected chain.

Cross-target helpers must still be accounted for.  Let $Y_1,Y_2\leftarrow\bits^\ell$ and retain $M=Y_1\oplus Y_2$.  If an execution for target 1 reaches a protected hidden point that returns $Y_2$, then $Y_1$ follows immediately.  Within one retained-state fiber, the helper response can vary over $2^\ell$ values, and its hidden-point address may vary as well.  Equation~\eqref{eq:online-sideinfo} therefore records the complete residual table-dependent foreign protected subtranscript.  In Contextual Chain, reaching such a point also requires constructing its branch-specific contextual prefix and protected predecessor: preparation carried across the cut is charged to $M_\Delta$, while genuinely new protected information obtained online is charged to $b_\Delta$.

\subsection{Deferred-decision representation of protected slices}\label{app:punctured}
The coupling in Eq.~\eqref{eq:hidden-point-coupling} is a standard-ROM deferred-decision representation.  On $\mathsf{GoodCtx}_A$, every protected transition has a distinct prefix $\rho_{i,s}$, so the corresponding restrictions of $H_{\rm rd}$ are independent random functions.  For each slice, the predecessor $P_{i,s}$ is fixed before that slice is sampled.  Sampling a full base function $G_{i,s}$ and then overwriting its value at $P_{i,s}$ with an independent uniform $Y_{i,s}$ therefore leaves the slice exactly random.

This representation has two proof advantages.  First, the finite vector $Y_A$ is independent uniform before the retained-state encoder is applied even though concrete protected addresses contain predecessor values.  Second, after the base functions are fixed, wrong candidate queries have deterministic answers, while a correct candidate hits a hidden point and exposes the corresponding protected chain value.  Conditioning on those base functions is proof-side bookkeeping, not operational free access.  Counterfactual executions retain their own depth and query budgets, and any table-dependent protected outputs obtained from unselected targets remain in the side-information transcript of Eq.~\eqref{eq:online-sideinfo}.

\subsection{Global causal-counting certificate}
The counting argument does not assume that protected outputs remain independent after the retained state is fixed.  Fix a target subset $A$, a retained-state value, one witness position per target, and one aggregate cross-target protected subtranscript.  Delete the $K$ selected-chain witness values and keep every other protected coordinate.  During deterministic replay, foreign correct hidden-point interactions are reproduced from the fixed subtranscript, wrong-point responses come from the fixed base functions, and undeleted selected-chain values come from the retained table.  Lemma~\ref{lem:first-timely-witness} ensures that replay reaches the certificate-specified witness exposure before the deleted witness is required.  The candidate produced at that exposure reconstructs the missing value.  Repeating this for all targets gives an injective description of length $D_A-K\ell$ protected bits plus witness positions, candidate ordinals, and the cross-target subtranscript, which yields the exponent $D_A-Ka_{\rm causal}+b_\Delta(A)$ in the main theorem.

\section{Source and Deployment Extensions}\label{app:deployment}
\subsection{Source, reconciliation, and active selection}\label{app:source-details}
A noisy deployment should distinguish raw observation attempts from accepted events. Write
\[
(E_n^A,E_n^B,A_n,\mathcal H_n)\leftarrow\mathsf{Rec}(W_n^A,W_n^B;a_n),
\]
where $A_n$ records acceptance and $a_n$ may represent active adversarial influence. If $N_t$ is the stopping time of the $t$-th accepted event, then the security source is $E_t=E_{N_t}^A=E_{N_t}^B$ on the synchronization-good event, not the raw $W_n$. Acceptance itself may leak substantial information, and selective jamming can become a source-selection attack. Accordingly the preselection view contains all helper transcripts, accept/reject history, timings, retries, aborts, and attacker observations.

Repeated reconciliation of the same physical sample is especially dangerous because helper leakage can accumulate; the default policy should use fresh observations for retries, or else evaluate all retry transcripts jointly. Event-level synchronization error should be promoted to an epoch-level bad-event probability. Exact distribution-free certification of a high-dimensional physical source from finite data is impossible without assumptions; a deployment should therefore use an explicit source family or confidence set and report mature coverage under the allowed attacker class.

Post-selection challenge/response tags are treated differently from zero-cost source leakage. They are public, but testing a candidate against them requires reconstructing the candidate readiness state. For delayed catch-up they are cost-bearing validators, not information to condition away for free.

\subsection{Physical resource interpretation}\label{app:physical}
Partition attacker storage into tiers $r$ with capacity $C_r$ and minimum random-access latency $L_r$. Define
\[
\mathcal R_\Delta=\{r:L_r\le\Delta\},\qquad C_{\rm RA}(\Delta)=\sum_{r\in\mathcal R_\Delta} C_r.
\]
Any successful physical realization of an abstract attack must provide at least the theorem's deadline-accessible retained-state length through tiers that can affect the response by the deadline. Thus $C_{\rm RA}(\Delta)\ge M_\Delta$ is a necessary systems-side condition. A slower tier contributes only through information prefetched across the selection cut into a faster reachable tier.

This does not imply a bandwidth lower bound of $M_\Delta/\Delta$: with a single random-lane challenge, an attacker may store a large indexed database yet read only one small record after selection. Multiple independently selected lane challenges could create a separate IOPS/bandwidth requirement, but the one-lane theorem does not rely on such a requirement. A full systems model is more naturally a feasible region in $(M,d,Q)$, with attacker-favorable bounds used for each coordinate.

\subsection{Persistent catch-up and wide-state continuity}\label{app:persistent}
One-epoch failure is not permanent loss. Let $\pi_{k:n}$ denote unresolved history conditioned only on zero-cost future source leakage, and let $\mathcal G_{k:n}$ denote the reconstruction graph containing context updates, readiness updates, historical validators, and state transitions. Define a dynamic reconstruction coverage $\operatorname{DRC}_W(\pi_{k:n},\mathcal G_{k:n})$ as the maximum posterior mass on which an adaptive strategy can catch up with cumulative work at most $W$, and
\[
W_p^\star(k,n)=\inf\{W:\operatorname{DRC}_W(\pi_{k:n},\mathcal G_{k:n})\ge p\}.
\]
Comparing $W_p^\star$ with the adversary's cumulative reconstruction budget expresses the maturation-versus-catch-up race without replacing posterior mass by a raw candidate count.

The short-state epoch reset in the basic construction can erase wide readiness debt once the correct branch has eventually been reconstructed. A persistent profile can instead carry readiness lanes across epoch boundaries, for example
\[
X_{k+1,0,j}=H_{\rm rd}(\texttt{EPOCH},k+1,\rho_{k+1},S_{k+1},j,X_{k,T,j}).
\]
The honest asymptotic state remains $O(J\ell)$, while an adversary that catches up must reconstruct the wide vector rather than merely a short digest.  This profile complements the one-epoch theorem when long-term continuity across repeated outages is required.

\subsection{Snapshot compromise and contextual healing}\label{app:healing}
Suppose an adversary obtains a complete snapshot $\Sigma_{t_{\rm rel}}$ and then loses direct state access. No security is claimed at the compromise instant. Subsequent accepted context is evaluated through
\[
\pi_v^{\rm heal}(e)=\Prb[\mathbf E_\tau^{\rm heal}=e\mid \Sigma_{t_{\rm rel}},V_{\rm pre}^{\rm heal}=v].
\]
The compromised state is therefore free information; all security credit comes from later residual contextual asymmetry. Once that uncertainty matures, the ordinary one-epoch causal theorem applies without modification. If $R_{\rm CC}^{\rm heal}(p,t)$ denotes the resulting retained-state requirement and $C_A(\Delta,t)$ the attacker's deadline-reachable capacity, a natural healing time is
\[
T_{\rm heal}(p)=\inf\{t>t_{\rm rel}:R_{\rm CC}^{\rm heal}(p,t)>C_A(\Delta,t)\}-t_{\rm rel}.
\]
Healing is an epoch-specific re-emergence of preparation asymmetry; permanent exclusion still depends on persistent catch-up. If a clone keeps the complete state and continues receiving exactly the same future contextual inputs, the executions remain identical and there is nothing for a state-and-context-only mechanism to distinguish.

\section{Calibration of \texorpdfstring{$d_\Delta$, $Q_\Delta$, and $b_\Delta$}{dDelta, QDelta, and bDelta}}\label{app:calibration}
For a concrete readiness primitive $F(c,x)$, $d_\Delta$ is measured with a genuinely dependent chain and $Q_\Delta$ with independent target-bearing predecessor-candidate evaluations, plus the single final output attempt. Generic library hash throughput is not an adequate substitute. Attack-side prefix precomputation, fixed-format kernels, SIMD/GPU batching, and hardware-specific layouts are permitted. Security evaluation should use attacker-favorable measurements or explicit acceleration margins. Network latency is not normally credited to security; the attacker is credited with the full verifier deadline $\Delta$.

The two hardware-rate parameters influence the theorem differently. Larger $Q_\Delta$ directly reduces the causal width, while larger $d_\Delta$ both reduces the causal width and moves the maturity frontier toward the past:
\[
a_{\rm causal}=[\ell-\log_2(1+Q_\Delta)-\log_2(d_\Delta+1)]_+,\qquad
\tau_\Delta=\max\left\{0,T-\left\lfloor d_\Delta/L\right\rfloor\right\}.
\]
Thus attacker sequential speed erodes both per-target protected width and the amount of contextual history eligible for security credit.

The quantity $b_\Delta(A)$ is different: it is not a raw hash-throughput benchmark but a proof-side log-cardinality of residual table dependence in the complete foreign protected query/response subtranscript after the retained state is fixed.  A memory-only numerical claim therefore needs an independent upper bound on $\bar b_\Delta(K)$.  The baseline $\bar b_\Delta(K)=0$ means that the foreign protected subtranscript has no residual table dependence once the retained state and proof environment are fixed; it does not forbid foreign queries.  More generally, the theorem can be reported as a joint retained-state/online-side-information tradeoff, with any conservative upper bound on $\bar b_\Delta(K)$ subtracted explicitly.  The numerical examples below label the zero-side-information baseline accordingly.

The accompanying code includes theoretical parameter sweeps; empirical hardware values must replace illustrative $d_\Delta$ and $Q_\Delta$ before platform-specific claims are made.

\section{Numerical Details and Reproducibility}\label{app:numerics}
All values and plots in this appendix are regenerated by the code included with the bundle.

\subsection{Exact binary coverage}

For a binary IID event source, let the more likely symbol have probability
\begin{equation}
q=2^{-\eta},
\qquad q\in[1/2,1].
\end{equation}
A length-$T$ sequence containing $k$ occurrences of the less likely symbol has probability
\[
q^{T-k}(1-q)^k
\]
and multiplicity $\binom{T}{k}$. Exact top-mass coverage can therefore be computed by accumulating probability classes in descending order, without enumerating $2^T$ sequences.

For $p=0.9$:
\begin{center}
\begin{tabular}{rrrr}
\toprule
$\eta$ & $T$ & $\log_2(p2^{\eta T})$ & exact $\log_2K_{0.9}$\\
\midrule
0.10 & 40  & 3.848 & 18.329 \\
0.10 & 100 & 9.848 & 43.365 \\
0.25 & 40  & 9.848 & 28.377 \\
0.50 & 40  & 19.848 & 36.266 \\
1.00 & 40  & 39.848 & 39.848 \\
\bottomrule
\end{tabular}
\end{center}
For $\eta=0.25$, $T=40$,
\begin{equation}
K_{0.9}=348{,}697{,}617,
\end{equation}
and the exact mass of the top $2^{20}$ sequences is approximately
\begin{equation}
0.3841131573.
\end{equation}
The difference between exact coverage and the min-entropy lower bound illustrates why posterior coverage remains the preferred operational quantity.  Figure~\ref{fig:binary-coverage} shows how the number of histories needed to cover 90\% probability grows with the event count and source rate.  The curves are distinguished by line style and marker shape so the ordering remains visible without color.

\begin{figure}[H]
\centering
\includegraphics[width=0.78\linewidth]{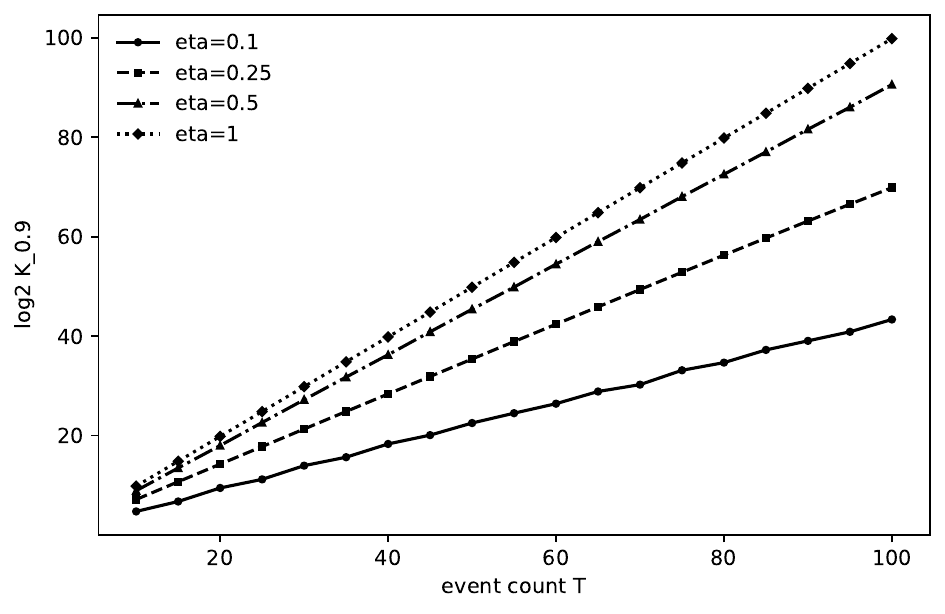}
\caption{Exact $\log_2K_{0.9}$ for controlled IID-binary sources.  Solid-circle, dashed-square, dash-dot-triangle, and dotted-diamond curves correspond to $\eta=0.10,0.25,0.50,1.00$.  For skewed sources, exact high-probability coverage can be much larger than the min-entropy count.}
\label{fig:binary-coverage}
\end{figure}

\subsection{Theorem-backed uniform-core causal example}

Consider a uniform mature core with $2^{20}$ contextual branches and $J=256$ lanes. Let
\[
\ell=256,\qquad
w=J\ell=8\ \mathrm{KiB},\qquad
Q_\Delta=2^{40},\qquad
d_\Delta=32,
\qquad p=0.9.
\]
For the baseline $\bar b_\Delta(K)=0$ case, the causal width is
\begin{equation}
a_{\rm causal}
=256-\log_2(1+2^{40})-\log_2 33
\approx210.956\ \text{bits}.
\end{equation}
The finite target core has
\[
N=J2^{20}=2^{28}
\]
branch--lane targets. Choosing the admissible value
\[
\beta=0.899955,\qquad
K_\beta=241{,}579{,}831,
\]
gives the same causal width $a_{\rm causal}\approx210.956$ bits for every covered target.  With $\bar b_\Delta(K)=0$, the resulting certified lower bound is
\begin{equation}
\boxed{
M_\Delta\gtrsim5.92\ \mathrm{GiB}.}
\label{eq:uniform-causal-example}
\end{equation}
The normalized bound is approximately $0.7398$ of $N\ell=w2^{20}$. This is the zero-side-information baseline; a nonzero upper bound on online protected side information subtracts directly from the retained-state lower bound.

Figure~\ref{fig:causal-storage-q} isolates the online-query tradeoff.  Increasing $Q_\Delta$ enlarges the set of candidate witness values that can be tried after selection and therefore reduces the state that must have crossed the selection cut.  The monotone decline is the intended calibration behavior: faster or more massively parallel online hardware must be compensated by a wider readiness state, a shorter deadline, or more mature contextual coverage.

\begin{figure}[H]
\centering
\includegraphics[width=0.78\linewidth]{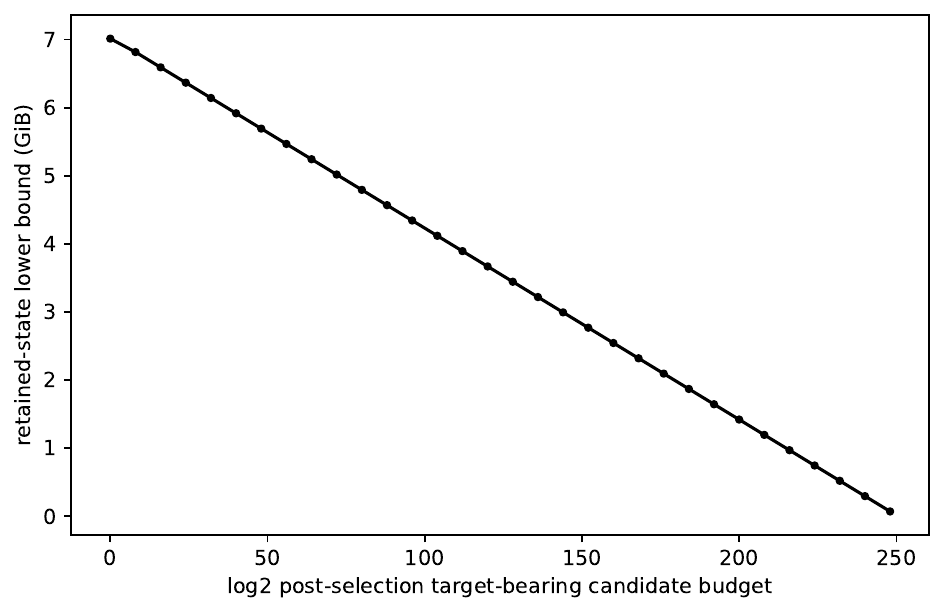}
\caption{Erosion of the zero-side-information retained-state bound as the selected-target candidate budget increases, for a uniform $2^{20}$-branch core, $J=256$, $\ell=256$, $d_\Delta=32$, and $p=0.9$.}
\label{fig:causal-storage-q}
\end{figure}

\subsection{Exact-coverage binary mature-source example}

The controlled IID-binary source can be combined consistently with maturation. Take an epoch of $T=40$ events, per-event readiness depth $L=4$, and post-selection sequential budget $d_\Delta=32$. Then
\[
\tau_\Delta
=40-\left\lfloor\frac{32}{4}\right\rfloor
=32.
\]
For $\eta=0.25$ bit/event, the exact 99\% mature-history core at length 32 contains
\begin{equation}
K_{0.99}=99{,}679{,}797
\end{equation}
contextual histories (the included mass is slightly above 0.99 because the last probability class is discrete). With $J=256$, the finite target core contains
\begin{equation}
N=25{,}518{,}028{,}032
\end{equation}
branch--lane targets. Using $\ell=256$, $Q_\Delta=2^{40}$, $d_\Delta=32$, and target online success $p=0.9$, the admissible choice $\beta=0.889956$ gives $K_\beta=1{,}733{,}947{,}800$ and $a_{\rm causal}\approx210.956$ bits. Exact target coverage in Eq.~\eqref{eq:partial-causal-bound} then gives
\begin{equation}
\boxed{
M_\Delta\gtrsim41.52\ \mathrm{GiB}.}
\label{eq:binary-causal-example}
\end{equation}
The large difference from the uniform $2^{20}$-branch illustration comes from exact mature coverage: for this skewed source, high-probability coverage can be far larger than the min-entropy count.  Figure~\ref{fig:binary-causal-tau} shows the second design lever: as more of the history becomes mature, exact coverage expands and the required retained state rises sharply.

\begin{figure}[H]
\centering
\includegraphics[width=0.78\linewidth]{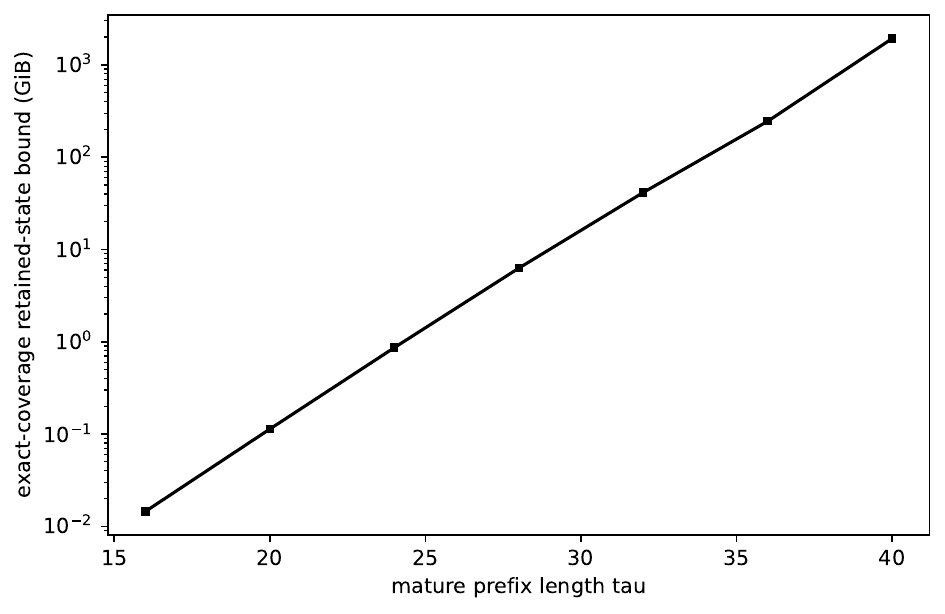}
\caption{Exact-coverage causal retained-state bound for the controlled binary source as the mature prefix length changes. Parameters are $\eta=0.25$, 99\% mature-history core, $J=256$, $\ell=256$, $Q_\Delta=2^{40}$, $d_\Delta=32$, and $p=0.9$.}
\label{fig:binary-causal-tau}
\end{figure}

\subsection{Maturation example}

For a pointwise source rate $\eta=0.25$ bit/event, readiness depth $L=4$, and varying post-selection sequential budget $d_\Delta$, Eq.~\eqref{eq:mature-rate} gives the mature min-entropy lower bound shown in Fig.~\ref{fig:maturity}. Recent events inside the regeneration horizon do not contribute to this lower bound.

\begin{figure}[H]
\centering
\includegraphics[width=0.78\linewidth]{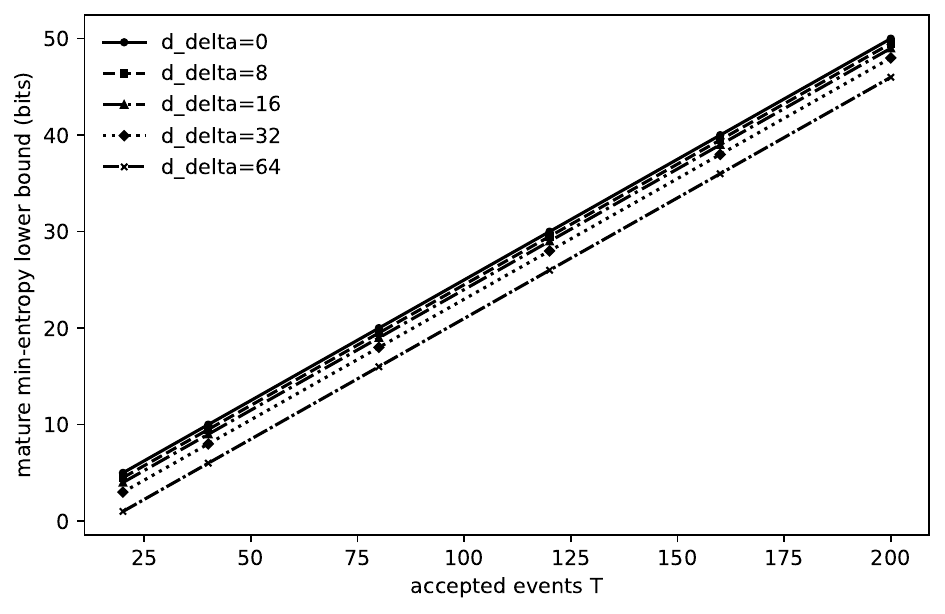}
\caption{Illustrative mature min-entropy lower bound $\eta\tau_\Delta$ with $\eta=0.25$ and $L=4$.  Distinct black line styles and markers denote $d_\Delta=0,8,16,32,64$; larger post-selection depth shifts the credited mature history downward.}
\label{fig:maturity}
\end{figure}

\subsection{Core-based state-width example}

In the same binary source with $\eta=0.25$, $T=40$, the 90\% source core has
\[
K_{0.9}=348{,}697{,}617.
\]
Using the full $T=40$ contextual updates required by $\mathsf{GoodCtx}$ and requiring the random-oracle fold-collision union bound to be below $2^{-128}$ gives the requirement
\begin{equation}
\lambda\ge190\ \text{bits}.
\end{equation}
This example illustrates why core coverage, rather than min-entropy alone, should control fold-state sizing.

\subsection{Streaming feasibility envelope}

Equation~\eqref{eq:streaming-envelope} gives, for example, $\tau_{\rm epoch}>25.6$ s when $J=256$, the adversarial sequential speedup factor is $s=100$, and $\Delta=1$ ms. This number is an algebraic illustration only; no hardware rate is assumed or measured.

The accompanying reproducibility bundle contains the exact coverage calculations, maturity and fold-collision sizing, parameter sweeps, all manuscript figures, and a validation script for the headline numerical values.  These calculations are theoretical illustrations unless an entry is explicitly labeled as an empirical benchmark.

\end{document}